\documentclass{article}
\usepackage{graphicx}
\usepackage{amssymb, amsmath, amsfonts, amsthm, xypic}
\usepackage{color}

\usepackage{palatino}

\usepackage{hyperref}

\usepackage[sc]{mathpazo}
\usepackage[T1]{fontenc}

\theoremstyle{plain}
\newtheorem{thm}{Theorem}[section]
\newtheorem{lem}[thm]{Lemma}
\newtheorem{prop}[thm]{Proposition}

\theoremstyle{definition}

\newtheorem{examp}{Example}[section]
\newtheorem{remark}{Remark}[section]

\DeclareMathOperator{\rank}{\textrm{rank}}
\DeclareMathOperator{\pr}{\textrm{pr}}
\DeclareMathOperator{\SM}{\textrm{SM}}
\DeclareMathOperator{\U}{\textrm{U}}
\DeclareMathOperator{\SU}{\textrm{SU}}
\DeclareMathOperator{\SO}{\textrm{SO}}
\DeclareMathOperator{\Sp}{\textrm{Sp}}
\DeclareMathOperator{\vac}{\textbf{Vac}}

\DeclareMathOperator{\NPC}{\mathsf{NPC}}
\DeclareMathOperator{\Poly}{\mathsf{P}}
\DeclareMathOperator{\NP}{\mathsf{NP}}

\newcommand{\Z}{\mathbb{Z}}
\newcommand{\R}{\mathbb{R}}
\newcommand{\g}{\mathfrak{g}}
\newcommand{\h}{\mathfrak{h}}

\title{Ranks of gauge groups and an NP-complete problem on the Landscape}
\author{Abhijnan Rej\\\small{Tata Consultancy Services Innovation  Labs}\\\small{1 Software Units Layout, Madhapur}\\\small{Hyderabad 500081, Andhra Pradesh, India}\\
\\
\small{{\tt abhijnan.rej@gmail.com, abhijnan.rej@tcs.com}}}
\date{}
\begin{document}
\maketitle
\begin{abstract}
We prove that the problem of determination of factor gauge groups given the rank of the gauge group at any given vacuum in the Landscape is in the computational complexity class $\NPC$. This extends a result of Denef and Douglas on the computational complexity of determination of the value of the cosmological constant in the Landscape.
\end{abstract}
\section{Introduction}
The purpose of this short note is to prove a result we first conjectured in \cite{Rej} regarding the computational complexity of the problem of determining the factor gauge groups given the rank of a gauge group at an arbitrary point in the Landscape of string vacua whose size, as estimated in \cite{AcharyaandDouglas}, is of the order $10^{500}$. This line of thought originates in the seminal work of Denef--Douglas \cite{DenefandDouglas} where it is shown that the problem of determining the value of the cosmological constant in the Bousso--Polchinski model at an arbitrary point in the Landscape is NP-complete. Indeed we find another NP-complete problem on the Landscape related to the determination of the factor gauge groups (theorem \ref{thm:main}.) 
\par
The note is organized in the following way: in the rest of this section, we provide a very brief introduction to both the problems of determination of ranks of gauge groups in the Landscape as well as that of computational complexity. The main result is described and proved in section \ref{sec:main} where the key idea is to reduce the rank determination problem to a known $\NPC$ problem. By way of a conclusion in \ref{sec:conc} we end with a brief remark on integer partitions and their possible role in the enumerative combinatorial aspects of the string Landscape. In an appendix \ref{sec:subsum}, we briefly describe the \emph{subset-sum problem} in computation theory, a key ingredient in our proof. 
\subsection{Ranks of gauge groups and the Landscape}
A interesting parameter is the statistical study of the Landscape of string vacua $\vac$ is the \emph{average rank} of a point in $\vac$. This is defined \cite{KumarandWells} in terms of the number of complex moduli of the compactified space and the flux expressed in terms of the configuration of D-branes wrapping it. An interesting fact about the analysis of average rank carried out by Kumar and Wells \cite{KumarandWells} is the absence of any parameter that depends explicitly on the structure of the compactified space (a Calabi--Yau 3-fold in their case.) In their analysis, when the type IIB mirror of the compactified space is the orientifold $\mathbb{T}^6/\Z_2$, the average rank is $\frac{16}{5}$ and in presence of a small cosmological constant exactly 4 which is the Standard Model gauge group rank.
\par
In \cite{GmeinerEtAl}, it is estimated that the frequency of occurrence of MSSM in the Landscape with SUSY intersecting D-branes on an toriodal orbifold compactified space is around $10^{-9}$. Dienes \cite{Dienes} has performed an exhaustive study on the \emph{heterotic } Landscape  focusing on the subclass of $10^5$ vacua. Of particularly interest for our study is the following result of his: let $f$ denote the number of irreducible gauge group factors. In figure 2 of \cite{Dienes}, the number of distinct heterotic string models with $f$ gauge groups factors is plotted as a function of $f$, the main upshot being that out of $10^5$ distinct models, only 1301 distinct gauge groups are obtained. (In what follows, Dienes' $f$ will be denoted as $m$.) Furthermore, his study of the 10-dimensional heterotic models reveal very interesting statistics about the average number of $\SO$, $\SU$ and exceptional groups factors.
\par
The premise of this note is to be agnostic about the nature of the (compactified) backgrounds in the Landscape, their explicit geometry, the geometry of D-branes configurations wrapping them or even the probability distributions of values of cosmological constant taken over the Landscape, allowing for various "exotic" possibilities. This level of generality is inspired by John Wheeler's famous dictum that \emph{everything that is not explicitly forbidden is allowed}. The only constraint on our analysis is that we have considered simply-laced compact Lie groups as possible gauge groups factors as opposed to the more general non-simply-laced or noncompact types. 
\subsection{Complexity classes $\Poly$, $\NP$ and $\NPC$}
Computational complexity theory, very roughly speaking, is the study of how fast or slow it is to implement an algorithm to solve a problem given the size of the algorithm input (in other words, the \emph{complexity} of the problem.)  A different but equivalent way of stating this is to say that computational complexity theory is a quantitative theory of efficiency of algorithms. As a mature branch of theoretical computer science, it uses many combinatorial and graph-theoretic results and underlies almost all of modern cryptography. It has also recently found applications in various branches of theoretical physics, cf. \cite{aaronson} for an interesting overview. 
\par
Depending on the complexity, problems can be classified into three broad classes $\Poly$, $\NP$ and $\NPC$. (The canonical introduction to computational complexity remains chapter 34 of \cite{CormenEtAl}.) The easiest class of problems lie in the class $\Poly$; these are problems that can be solved in \emph{polynomial} time (polynomial in the length of the input). The next class of problems are ones in $\NP$; these are problems whose solutions can be verified to be true in polynomial time (the notation stands for \emph{nondeterministic polynomial} and {\bf not} "nonpolynomial".) Finally the hardest kind of problems are in the class $\NPC$ (\emph{nondeterministically polynomial complete}); these are problems for whom it is not even known if an efficient algorithm exist that solves it-- philosophically, problems in the class $\NPC$ are often thought of as "intractable". Generically speaking, it is not possible to obtain an optimally efficient solution to an $\NPC$-problem even though given a solution (\emph{the certificate}), it may be possible to verify the solution in polynomial time. The most famous of all $\NPC$ problems is the \emph{traveling salesman problem}. Other well-known problems in the class $\NPC$ include the \emph{graph isomorphism problem} (i.e. the problem of determining whether two given graphs are isomorphic or not) and the so-called \emph{subset-sum} problem which we will crucially use in our main result. (Appendix \ref{sec:subsum} gives a brief overview of the latter.) These three complexity classes have the following relationship with one another (cf. \cite{CormenEtAl} p. 1070):
\begin{eqnarray*}
& & \Poly \subseteq  \NP,\\
& & \NPC  \subset  \NP, \\
& & \NPC \cap \Poly  =  \emptyset.
\end{eqnarray*}
Convention holds that in the first equation, the equality is never attained. (Whether this is indeed true or not is the subject of the immensely important \emph{$\Poly \neq \NP$ conjecture}.)

\section{Main result}\label{sec:main}
\subsection{Ranks and maximal tori}
By a \emph{torus} $T$ of real dimension $r$ we mean a space \[\underbrace{\R/\Z \times \cdots \times \R/\Z}_{r-\textrm{times}}. \] The \emph{rank} of a compact connected Lie group $G$ is the real dimension $r$ of a torus $T$ embedded in it which is also a maximal (Lie)-subgroup of $G$. Such a torus is called a \emph{maximal torus}.
\par
Let $G$ and $G'$ be two connected compact Lie groups with associated maximal torus $T$ and $T'$ respectively. Then following result is well known and included for the benefit for the reader.
\begin{prop}\label{prop:torus}
The maximal torus associated to the products of the Lie groups $G \times G'$ is $T \oplus T'$.
\end{prop}
\begin{proof}
We will provide this statement through contradiction. We know that $T \oplus T'$ is a compact connected abelian subgroup of $G \times G'$. Therefore, it must be contained in some (possibly larger) maximal torus of $G \times G'$. Let that maximal torus be $S$. If it were strictly bigger, then there would be an element $(x,y) \in S$ that is not contained in $T \oplus T'$. Without loss of generality let us suppose that $x$ does not belong to $T$. Then the projection map 
\[ \pr: G\times G' \longrightarrow G\]
will map $S$ to a compact connected abelian subgroup of $G$ strictly containing $T$ (the image containing $x$ that does not belong to $T$). This contradicts the maximality of the torus $T$ in $G$.
\end{proof}
Using proposition \ref{prop:torus}, we prove the following key lemma.
\begin{lem}\label{lem:prodrank}
Let $G$ and $G'$ be compact connected Lie groups. Then
\[ \rank(G \times G') = \rank(G) + \rank(G').\]
\end{lem}
\begin{proof}
From proposition \ref{prop:torus} we know that the maximal torus associated to the product of the Lie groups $G \times G'$ is $T \oplus T'$. Therefore, by definition,
\begin{eqnarray*}
\rank(G \times G') & = & \dim(T \oplus T'), \\
                   & = & \dim(T) + \dim(T'),\\
                   & = & \rank(G) + \rank(G').
\end{eqnarray*} 
\end{proof}
For a set $\{G_i\}$, $1 \leq i \leq n$, of compact connected Lie group with an associated sequence of maximal tori $\{T_i\}$, $1 \leq i \leq n$, the above proposition \ref{prop:torus} and lemma \ref{lem:prodrank} generalizes \emph{in toto}:
\[ \rank(G_1 \times \cdots \times G_n) = \rank(G_1) + \cdots + \rank(G_n).\]
\begin{examp}
The Standard Model of particle physics is given by a product of compact connected Lie groups $\U(1)$, $\SU(2)$ and $\SU(3)$ (the \emph{factor gauge groups}) as
\[ G_{\SM} = \SU(3) \times \SU(2) \times \U(1),\]
so 
\begin{eqnarray*}
\rank(G_{\SM}) & = & \rank(\SU(3)) + \rank(\SU(2)) + \rank(\U(1)),\\
             & = & (3-1) + (2-1) + 1 = 4.
\end{eqnarray*}
(The special unitary groups $\SU(n)$ have rank $n-1$ and the unitary group $U(n)$ has rank $n$.)
\end{examp}
Consider the Landscape of string vacua $\vac$ and suppose we want to determine the factor gauge groups on a single point $\mathcal{U} \in \vac$ given that the rank of the gauge group products $\mathcal{G}$ is known and fixed:
\[ \rank(\mathcal{G}): = \alpha.\]
An alternative formulation: let $\vac$ be the set of inequivalent string vacua. Define the function
\begin{eqnarray*}
\rank : \vac & \longrightarrow & \Z^{+},\\
	\mathcal{U} & \mapsto & \rank({\mathcal{G}}).
\end{eqnarray*}
which maps a vacuum $\mathcal{U}$ to the rank of the gauge group, $\rank({\mathcal{G}})$, at $\mathcal{U}$. In general, the rank function is not injective. Our central claim will be the following: the problem of determining the factor gauge groups given a ("target") value of the rank function lies in the complexity class $\NPC$.
\par
To see this, let us assume $\{G_j\}$, $1 \leq j \leq m$, to be a set of Lie groups that hypothetically constitute the $m$ factors of $\mathcal{G}$ (The number $m$ is often referred to in the Landscape literature as the \emph{shatter}.) These factor groups can have any structure as long as they are \emph{simply-laced, connected and compact and each should arise as a single factor in $\mathcal{G}$}-- the set of (isomorphism classes of) such groups is finite, following Cartan's classification. Let \[\rank({G_j}) = \alpha_j.\] We do not require that $\alpha_k = \alpha_l \implies k = l$.
\par
We know that
\[ \rank(G_1 \times \cdots \times G_m) = \sum_{j = 1}^m \alpha_j\]
by lemma \ref{lem:prodrank}. 
\begin{thm}\label{thm:main}
Let $S$ be a set of positive integers and $\alpha$ be a fixed positive integer. The problem of determining whether $S$ has a subset $S' = \{\alpha_j\}_{j=1}^m$ with \[ \sum_{j = 1}^m \alpha_j = \alpha \] is in the complexity class $\NPC$.
\end{thm}
\begin{proof}
The problem as formulated is exactly the famous \emph{subset-sum problem}. For a proof that this problem is in the complexity class $\NPC$ see theorem 34.15 of \cite{CormenEtAl}, p. 1097. 
\end{proof}
In computational complexity theory $\alpha$ is called the \emph{target value}; in appendix \ref{sec:subsum}, we discuss the subset-sum problem at an elementary level for readers unfamiliar with it.
\par
It is very important to note that theorem \ref{thm:main} holds true (and indeed, any application of computational complexity is permissible to string theory!) only if $|\vac| < \infty$. The fact that this is true is the extremely famous argument of Acharya--Douglas \cite{AcharyaandDouglas}.
\begin{remark}
It is important to keep in mind what these two sets $S$ and $S'$ correspond to, physically. While $S'$ corresponds to the set of all simply-laced, connected and compact gauge groups that hypothetically contribute single factors to $\mathcal{G}$, $S$ corresponds to a much larger set of all not necessarily simply-laced Lie groups and products of them (the number of factors in products of such groups being bounded by some very large integer) such as $\SO(2n+2)$, $\Sp(2,n)$, $F_4$, $E_8$ and $G_2$ (the last one is is a 7-dimensional Riemannian manifold with has holonomy group as the octonions.) The set $S$ then is the set of ranks of all such groups (also accounting for multiplicities that arise out of taking products). As Dienes notes \cite{Dienes}, even in the context of the Landscape of 10-dimensional heterotic string models (which is a minuscule part of the whole Landscape), allowing for generality on the gauge groups that can arise forces us to also take into account these non-simply-laced models. Taking into account $\mathcal{N} = 1,2 \textrm{ or } 3$ SUSY can significantly increase the size of $S$.
\end{remark}
\subsection{The Lie algebras picture}
By the term "rank", a practicing physicist often simply mean the maximum number of linearly independent commuting elements of a Lie algebra $\g$, in other words, the dimension of the maximal abelian subalgebra of $\g$. (A Lie group $G$ and its associated Lie algebra $\g$ are functorially related to each other through the surjective \emph{exponential map} $\g \longrightarrow G$.) 
\par
In the next few lines we remind the reader why this definition in terms of a maximal abelian sub-Lie algebra is the same as the definition in terms of maximal torus associated to $G$. Recall a basic result in Lie theory (\cite{Bump}, p. 120) that uses Cartan's theorem in an essential way:
\begin{prop}\label{prop:maxabtor}
Let $G$ be a compact connected Lie group with Lie algebra $\g$. A maximal abelian subalgebra $\mathfrak{h}$ of $\g$ is the Lie algebra of a conjugate of a maximal torus $T$ of $G$. The dimension of $\mathfrak{h}$ is the rank of $G$.
\end{prop}
Let $G$ and $G'$ be compact connected with maximal tori $T$ and $T'$ and with associated maximal abelian subalgebras $\h$ and $\h'$ respectively. By lemma \ref{lem:prodrank}, the rank of $G \times G'$ equals the sum of ranks of $G$ and $G'$. Applying proposition \ref{prop:maxabtor} to lemma \ref{lem:prodrank}, we see that this is the sum of the dimensions of $\h$ and $\h'$ and therefore our main theorem \ref{thm:main} goes through in exactly the same way with this equivalent (and perhaps more familiar) notion of rank.

\section{Integer partitions and gauge groups ranks?}\label{sec:conc}

There is another interesting combinatorial problem associated to the determination of factor gauge groups given the gauge group rank at a point in a Landscape, namely considering its integer \emph{partitions}.  Recall that the (arithmetic) partition function $p(n)$ counts the number of ways $n$ can be written as a sum of positive integers. Explicit number-theoretic bounds and other properties of $p(n)$ form a rich part of arithmetic theory, cf. chapter XIX of \cite{HardyandWright}. A possible way of thinking of theorem \ref{thm:main} is in terms of partitions of the gauge group rank $p(\alpha)$ in terms of increasing factor gauge group ranks. An explicit asymptotic form of partition function (essentially due to Euler) goes as
\[ p(\alpha) \sim \frac{\exp(\pi \sqrt{2 \alpha/3})}{4 \alpha \sqrt{3}}.\]
It can be seen (p.361 of \cite{CormenEtAl}) that this is less that $2^{\alpha-1}$ but still grows faster that an polynomial in $\alpha$ (so the problem of computing all increasing integer partitions of $\alpha$ is provably not in the complexity class $\Poly$.) It would be an interesting task to recast the gauge groups determination problem as presented in this note in terms of the rich arithmetic of the partition function $p(\alpha)$.
\vspace{5mm}
\subsection*{Acknowledgements} We thank Snigdhayan Mahanta for many pertinent discussions on Lie theory and Matilde Marcolli for her interest in this work, with the disclaimer that all remaining errors in this note is ours alone. We also thank our employer Tata Consultancy Services for allowing us to devote part of our time to fundamental research.

\appendix
\section{The subset-sum problem}\label{sec:subsum}
Our exposition follows \cite{CormenEtAl}, 34.5.5, p. 1097.  The basic setup of the subset-sum problem is a finite set of integers $S$ and a fixed integer--the target value-- $t$ and the goal is to determine whether there is a subset $S'$ of $S$ such that \[ \sum_{s \in S'} s = t.\] (The reader may want to construct a few examples of when this is indeed the case.)
\par
The subset-sum problem, as a language, is defined as:
\[
{\tt SUBSET-SUM} = \left\{ <S,t> \quad :  \textrm{ there exists a subset }S' \subseteq S\textrm{ such that } t = \sum_{s \in S'} s.\right\}
\]
To prove that the problem is indeed in the class $\NPC$, we first have to show that it is in $\NP$. This is done using a verification algorithm with the desired subset $S'$ as a certificate. Proving that it is also $\NP$-complete is more involved and done using a formula satisfiability problem called 3-CNF-SAT. (This is a standard application of a \emph{reduction algorithm}, cf. figure 34.1 of \cite{CormenEtAl}.)
\par
The subset-sum problem can also be viewed as an \emph{optimization problem} as opposed to a \emph{decision problem} (as we have done.) An exponential-time exact algorithm and a fully polynomial-time approximation scheme to this problem is given in section 35.5 of \cite{CormenEtAl}.

\end{document}